\documentclass[10pt,letterpaper,margin=1in, conference]{ieeeconf}
\IEEEoverridecommandlockouts                            
\usepackage{color} 
\usepackage{amsfonts}
\usepackage{graphicx} 
\usepackage{latexsym,amsmath} 
\usepackage{mathrsfs} 
\usepackage{booktabs}
\usepackage{multicol}
\usepackage[makeroom]{cancel}
\usepackage{soul,xcolor}

\usepackage{amsthm}

\renewcommand{\baselinestretch}{0.93}

\title{Decentralized Frequency Control using Packet-based Energy Coordination$^\ast$\thanks{$^\ast$This work was supported by the U.S. Department of Energy's ARPA-E award DE-AR0000694.}}

\author{Hani Mavalizadeh
    $\;\;\;$ Luis A. Duffaut Espinosa$^\dagger$\thanks{$^\dagger$ This author was also supported by NSF grant CMMI-1839387.}
    $\;\;\;$ Mads R. Almassalkhi $^\ddag$\thanks{$^\ddag$ M. Almassalkhi is co-founder of startup company Packetized Energy, which is commercializing aspects related to PEM.} 
    \thanks{The authors are with the Department of Electrical and Biomedical Engineering, University of Vermont, Burlington, VT 05405 USA.} 
}

\allowdisplaybreaks[4]

\graphicspath{{figures/}}

\newtheorem{theorem}{Theorem}
\newtheorem*{remark}{Remark}
\newtheorem{definition}{Definition}
\newtheorem{assumption}{Assumption}
\theoremstyle{definition}

\begin{document}
	\bstctlcite{IEEEexample:BSTcontrol}

	\maketitle
	
	\begin{abstract}
		This paper presents a novel frequency-responsive control scheme for demand-side resources, such as electric water heaters. A frequency-dependent control law is designed to provide damping from distributed energy resources (DERs) in a fully decentralized fashion. This local control policy represents a frequency-dependent threshold for each DER that ensures that the aggregate response provides damping during frequency deviations. The proposed decentralized policy is based on an adaptation of a packet-based DER coordination scheme where each device send requests for energy access (also called an ``energy packet'') to an aggregator. The number of previously accepted active packets can then be used a-priori to form an online estimate of the aggregate damping capability of the DER fleet in a dynamic power system. 
		A simple two-area power system is used to illustrate and validate performance of the decentralized control policy and the accuracy of the online damping estimating for a fleet of 400,000 DERs.   
	\end{abstract}
	
	
	\section{Introduction}\label{sec:intro}
	Due to environmental concerns and energy policy, the integration of inverter-based renewable energy sources (RES), such as wind and solar PV in power systems is increasing. However, the intermittent nature of these sources has introduced new challenges. One of these is the decrease in power system inertia, which is defined as the power system's ability to oppose changes in frequency \cite{ercot2018}. The kinetic energy stored in rotating machines such as synchronous generators and loads is the main source of inertia in conventional power systems. With decreasing costs of renewable generation, some of the conventional generators will not be competitive and, thus, will go offline, which results in increased renewable generation capacity on the grid and reduction in the system's inertia. This subsequently increases the rate of change of frequency (ROCOF) and makes primary frequency control more challenging. Another factor, affecting frequency response is frequency-dependent loads. This is usually modelled by a damping constant which is defined as the percentage change in total system load in response to frequency change~\cite{ercot2018}. One way to remedy this is to require alternative sources of damping and inertia.For example, during a sudden loss of generation some of the stored kinetic energy in synchronous machines is used to compensate for frequency deviations. 
	
	Using spinning reserves to compensate for variability in renewable generation is an expensive solution for power systems with high penetration of RES~\cite{brokish2009}. With the development in sensor technologies, sensor cost has been reduced considerably, which has enabled DERs to provide fast ancillary services, such as frequency control. To have effective primary frequency control, DERs should react within 1 to 2 seconds~\cite{Lundstrom2018PESGM}. Much work has been done to find an efficient and fast method to use loads endowed with local control policies to provide frequency regulation ~(\cite{ brokish2009,amini2016ISGT,wang2019SG,wang2019PES, dorfler2016ACC,boyd2011FTML,schweppe1980PAS,zhao2015ACC, nandanoori2018CTA, schweppe1980PAS, tinderman2015CST,zhang2012CDC}).
	
	    Different structural coordination models are used for DER coordination. In centralized control, all of the measurements are sent to a coordinator and then decisions are broadcast by coordinator to each individual device. Distributed control requires coordination from a centralized agent but includes local sensing, computation, and control. Fully decentralized control has no centralized coordinator and the system response is a function of only local sensing, computation, and actuation.

	 Fully decentralized frequency control enables devices to make local decisions during frequency disturbances by using a local control law for DERs and generators. A distributed secondary frequency control for multi-area systems is presented in~\cite{wang2019SG} which restores the nominal frequency while minimizing the regulation costs. ~\cite{wang2019PES} estimates virtual load demand for controllable generators based on local frequency but the approach does not consider demand side frequency control. ~\cite{ dorfler2016ACC} uses gather-and-broadcast and generalized continuous-time feedback control version of the dual decomposition method~\cite{ boyd2011FTML} to dispatch DERs during the frequency deviations.
	
	Different methods for prioritization of DERs are presented in the technical literature. In~\cite{schweppe1980PAS}, a frequency adaptive power-energy re-scheduler (FAPER) is introduced, which sets frequency-dependent thresholds based on the  temperature of temperature-bound appliances such as EWH or air conditioners. In~\cite{brokish2009}, those same frequency thresholds are randomized in order to overcome the problem of synchronization in FAPER. In~\cite{nandanoori2018CTA} local measurements are sent to the aggregator within 5 to 15 minutes and a fitness value is calculated for each DER. These fitness values are then used by the aggregator to assign frequency-response thresholds in a prioritized manner. While the method works fine for small aggregation of DERs, sharing local measurements from large-scale fleets of DERs with an aggregator in real-time incurs large communication costs.
	Another important aspect in the problem of regulation is to guarantee the stability of the system when using DERs for primary frequency control. In~\cite{zhao2015ACC} a fully decentralized integral control is employed. This approach, yields global asymptotic stability condition but it does not guarantee the results to be optimal or feasible for the economic dispatch problem. This deficiency is addressed by introducing a distributed averaging-based integral (DAI) control, which operates by sensing local frequency and allowing communication between neighboring system busses.
		
	Several packet-based energy management systems have been proposed to coordinate DERs~\cite{ zhang2012CDC, lee2011PES, bashash2019energies}. In packet based systems, loads request for energy packets with given duration and amount. Each device is equipped with a timer which turns the device OFF when the packet length is completed. In this paper, the focus is on packetized energy management (PEM), which is a packet-based approach. Most frequency regulation approaches using DERs rely on communication between DERs and the (centralized) aggregator. However, sending online measurements from large number of appliance to aggregator requires a reliable, low-latency communication network, which can be costs at scale. Therefore, a fully decentralized frequency control scheme that does not require real time communication with aggregator is desired. This requires adapting PEM to the appropriate timescale, which is the main contribution of this paper. Thus, this paper develops a novel, and fully decentralized primary frequency control scheme within the existing PEM literature. Specifically, each controllable load responds to a locally measured frequency deviation based on a  pre-determined local control law and timer state which begets fast response and provides damping. The distribution of timers is available to the aggregator since it knows how many devices are accepted during any time period. These data can be used to accurately predict the effect of any frequency deviation on the fleet, thus quantifying the damping available in real-time, which is valuable to grid operators. The proposed scheme can be applied to any packet-based energy management. The main contributions of this paper can be summarized as follows:
	\begin{itemize}
		\item A responsive and fully decentralized frequency control policy is designed within a packet-based energy management systems that automatically prioritizes resources based on local dynamic states. 
		\item The equivalent damping for a fleet of resources operating with PEM is estimated analytically and enables the aggregator to quantify, in real-time, the effect of packetized DERs on primary frequency control prior to any event. 
		\item Simulation-based analysis is used to validate the performance of the proposed local control policy and the online estimate of damping provided by the aggregator on a dynamic, two-area model under various operating conditions.
	\end{itemize}
	This paper is organized as follows: Section~\ref{sec:model}, provides the network model and PEM preliminaries. The describes the decentralized packet-based frequency control scheme in described in Section~\ref{sec:decentralized}. Section`~\ref{sec:predict} provides the method for quantification of PEM damping by the aggregator. In Section.~\ref{sec:results}, an illustrative simulation for two-area system is shown to verify performance and validity of the methodology. The conclusions are given in Section~\ref{sec:conclusion}.

	\section{Modeling Preliminaries}\label{sec:model}
	
	In this section, first the network model is presented. This is followed by a brief summary of the PEM scheme. Specifically, the concepts of packet duration and packet interruption are provided. These are at the core of the proposed decentralized frequency control scheme.
	
	\subsection{Network Model}
     Let $G=(\mathcal{V},\mathcal{E}$) be a graph representing the topology of a transmission network, where $\mathcal{V}$ is the set of $N$ system nodes $\mathcal{V}=\{1,...,N \}$, $\mathcal{E}$ is the set of edges between nodes ($\mathcal{E} \subseteq \mathcal{V}\times \mathcal{V}$). If a pair of nodes $i$ and $j$ are connected by a tie line, then the tie line $(i,j) \in \mathcal{E}$. Each node represents a frequency control area which is modeled with aggregate generation and loads for each area. The load consists of two components: 1) uncontrollable loads and 2) PEM-enabled loads. The dynamics of the angular frequency deviations from nominal, $\Delta \omega(t) := \omega(t)-\omega_0 $, in each area~$j$ is determined by the swing equations~\cite{kundur}:
	\begin{subequations} 
	\begin{align}  
   \: \Delta\dot{\theta}_j(t)=& \Delta\omega_j(t),\\
	M_{j}\Delta\dot{\omega}_{j}(t) =& \Delta P^{\text{G}}_{j}(t)-\Delta P^{\text{L}}_{j}(t)-\Delta P^{\text{PEM}}_{j}(t)-D_{j}	\Delta \omega_{j}(t) 	\notag\\
	&  + \sum_{i:(i,j) \in \mathcal{E}}^{N} B_{ij}\left(\Delta \theta_{i}(t)-\Delta \theta_{j}(t)\right),  \label{eq:swingDyn2}
	\end{align}
	\end{subequations}
	where $\Delta P^{\text{G}}_{j}$, $\Delta P^{\text{PEM}}_{j}$ and $\Delta P^{\text{L}}_{j}$ are the change in generation, PEM controlled loads, and uncontrollable loads at node $j$ after the occurrence of a disturbance, respectively. $D_{j}$ and $M_j$ are the damping provided by motor loads and generator rotational losses in [MW/Hz] and the power system inertia in [s]. The last term in~\eqref{eq:swingDyn2} describes the power flow between area $i$ and $j$. DC power flow is used to calculate the power flow as in~\cite{liu2011PESG}. Before the disturbance, the frequency is assumed to be $f_{0}=60$Hz, which in radians amounts to $\omega_0:=2\pi f_0$. The emphasis of this manuscript is on frequency response and active power changes, therefore it is reasonable to neglect bus voltages and reactive power for the sake of simplicity~\cite{zhao2014AC}.
	
    To model the generators' primary frequency response, the effect of the generators' turbine is included as follows
	\begin{align}\label{eq:turbine}
	\tau_j \Delta \dot{P}^{\text{G}}_{j}(t)=\Delta P^{\text{set}}_{j}(t)-\Delta P^{\text{G}}_{j}(t)-\frac{\Delta \omega_{j}(t)}{R_{j}}, 
	\end{align}
	where $\tau_j$ is the time constant [s] and $\Delta P^{set}_{j}$ is the change in generation set point [MW] compared to the nominal state and $R_{j}$ is the equivalent droop of all the generators in area $j$ in [Hz/MW]. During primary frequency control, $\Delta P^{set}_{j}$ does not change. Together,~\eqref{eq:swingDyn2} and~\eqref{eq:turbine}, model the frequency dynamics related to primary frequency control in a transmission network. This model will be used in Section~\ref{sec:results} to quantify the role of the proposed decentralized DER control policy. 
	
	\subsection{Packetized energy management}\label{subsec:pem}
	
	 	\begin{figure}
 		\centering
 		\includegraphics[width=1\linewidth]{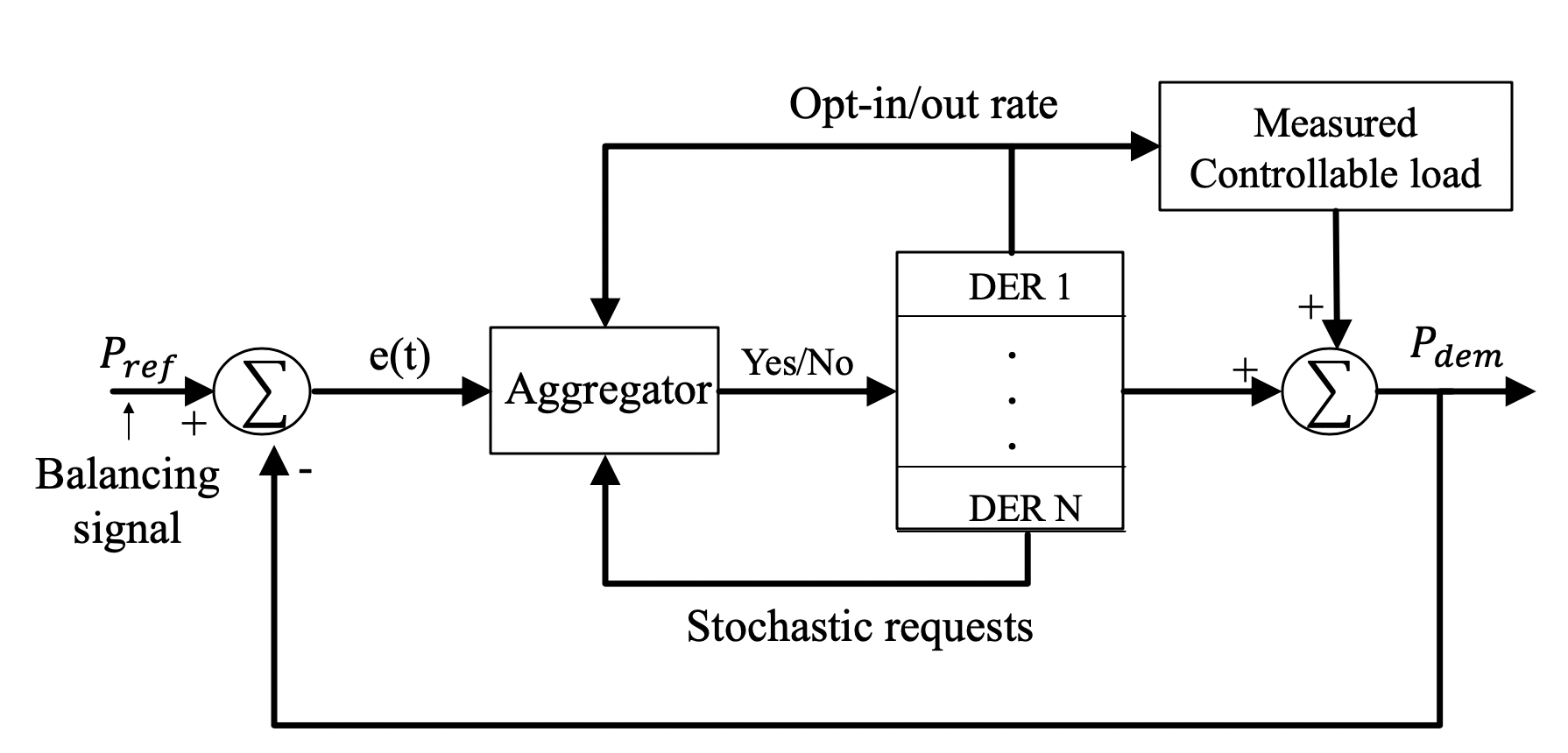}
 		\caption{The closed-loop feedback system for the packetized energy management with the reference signal $P_{ref}$ provided by the Grid Operator.}
 		\label{fig:blockdiagram}
 	\end{figure}
 	
    Previously, PEM has been used to provide load balancing and regulation grid services. In this paper, PEM capabilities are extended with a novel decentralized DER control policy that dynamically prioritizes which packets to interrupt to provide damping. This paper proposes a timer-based prioritization, but the underlying ideas also apply to other ``fitness''-based DER prioritization schemes (e.g.,~\cite{nandanoori2018CTA}). 
    
    For completeness, a brief description of PEM is provided next and illustrated in Fig.~\ref{fig:blockdiagram}. Under PEM, DERs request stochastically for energy packets from the grid based on their local dynamic state e.g. temperature for thermostatically controlled loads (TCL) or state of charge (SoC) for batteries and electric vehicles (EVs). This randomization limits synchronization between DERs. For example, an electric water heater with lower temperature is more likely to request an energy packet than devices with higher temperatures. The probability of requesting a packet at time-step k is modelled as a discrete-time variable as
    \begin{align}\label{eq:request}
	 P(T_{n}[k])=1-e^{-\mu (T_{n}[k])\Delta t},
	\end{align}
     where $T_n[k]$ is the DER's temperature and $\mu$ is a design parameter defined as~\cite{almassalkhi2018chapter}:
	\begin{align*}
	\lefteqn{\mu(T_{n}[k])=}  \\
	& \; \begin{cases} \mbox{0,} & \mbox{$T_{n}[k] \geq T^{\text{max}}$}\\ \mbox{$m_{R}(\frac{T^{\text{max}}_{n}-T_{n}[k]}{T_{n}[k]-T^{\text{min}}_{n}}) . \frac{T^{\text{set}}_{n}-T^{\text{min}}_{n}}{T^{\text{max}}_{n}-T^{\text{set}}_{n}}$,} & \mbox{$T^{\text{min}}_{n}<T_{n}[k]<T^{\text{max}}_{n}$}\\ \mbox{$\infty$,} & \mbox{$ T^{\text{min}}_{n} \geq T_{n}[k]$} \end{cases} \
	\end{align*} 
    with $m_{R}>0$ is the mean time-to-request (MTTR) and $T^{\text{max}}_{n}$, $T^{\text{min}}_{n}$, $T^{\text{set}}_{n}$ are maximum, minimum and setting temperatures, respectively.  
	
    The DERs' requests are then sent to the aggregator asynchronously and, based on grid or market reference tracking signal, are either accepted (e.g., DER is allowed to turn ON) or denied (e.g., DER is not allowed to turn ON). When an energy request is accepted, then a local timer is triggered and the DER will turn ON for a pre-defined period of time, called epoch length $(\delta)$. When the epoch length is completed, the device will turn OFF automatically and its local timer is reset. The local timer is defined as:
	\begin{align} \label{eq:timer}
	 t_{n}[k+1]= 
	        \left\{
        	 \begin{array}{ll} 
        	    t_{n}[k]+\Delta t,   & \text{if } C_n[k]= 1\\
        	    0                    & \text{otherwise}  
        	\end{array}    
           \right.,
	\end{align}
	where $\Delta t$ is the sampling time. The number of bins is $n_p :=\lfloor \frac{\delta}{\Delta t} \rfloor$. In particular, when the $n$-th DER has its request accepted at time step $k$ and is consuming its packet, then $C_n[l] = 1$ for all $l \in [k+1,k+n_p]$ and $C_n[k]= 0$, if request is denied. According to~\eqref{eq:timer}, when DER is in standby mode, $C_{n}$ and $t_n$ are both set to zero. Using the internal timer of each device, a distribution of timer states can be constructed at the aggregator that keeps track of the number of accepted requests at each time during the latest epoch. This is represented in~\cite{duffaut2017CDC} by 
	 \begin{align}  \label{eq:histogram}
    	 x[k+1] = Mx[k]+Bq^{+}[k],
        \end{align}
    where $x \in \mathbb{R}^{n_p\times 1}$ is the vector of number of DERs in each bin, $q^{+} \in \mathbb{R}$ is the number of accepted requests at the current time, $B \in \mathbb{R}^{n_p \times 1}$ is responsible for allocating the new accepted requests into the first bin. That is, $B$ is a zero matrix except for its first element whose value is 1. $M \in \mathbb{R}^{n_p \times n_p}$ is a zero matrix except for its first lower diagonal whose components are 1. The number of devices completing their packets at time-step $k+1$ is equal to the number of devices in the last bin of the timer distribution, $x_{n_p}[k]$. Consider using PEM for tracking a slowly moving reference signal and the frequency is maintained close to the nominal frequency. Fig.~\ref{fig:localtimer} shows the distribution of timer states ($t_n[k]$) and temperatures ($T_n[k]$) for devices $n$ that are ON at time-step $k$ for an example system of 400,000 EWHs. The number of bins for temperature and timer status is 30 and 10, respectively and $\delta$=3 minutes. The histogram of timers in~Fig.~\ref{fig:localtimer} is constructed using~\eqref{eq:timer} and~\eqref{eq:histogram} while the histogram of temperatures must be estimated. Note that some of the devices with high $t_n$ have low temperatures because of recent high water usage. In addition, some of devices with high temperature have just started their timer. This is caused by the random nature of PEM requests.  The clear correlation between temperature and timer occurs because the more energy that a device consumes, the higher the temperature. For simplicity, The DER capacities are assumed homogeneous. The dashed line shows the average number of devices in each bin, which is
	\begin{align}\label{eq:xbar}
	\bar{x}[k]=\frac{1}{n_p} (\mathbf{1}^{\top}_{n_p} x[k]).
	\end{align}
	where $\mathbf{1}_{n_p}\in\mathbb{R}^{n_p}$ is vector of ones.
	
	\begin{figure}
		\centering
		\includegraphics[width=1\columnwidth]{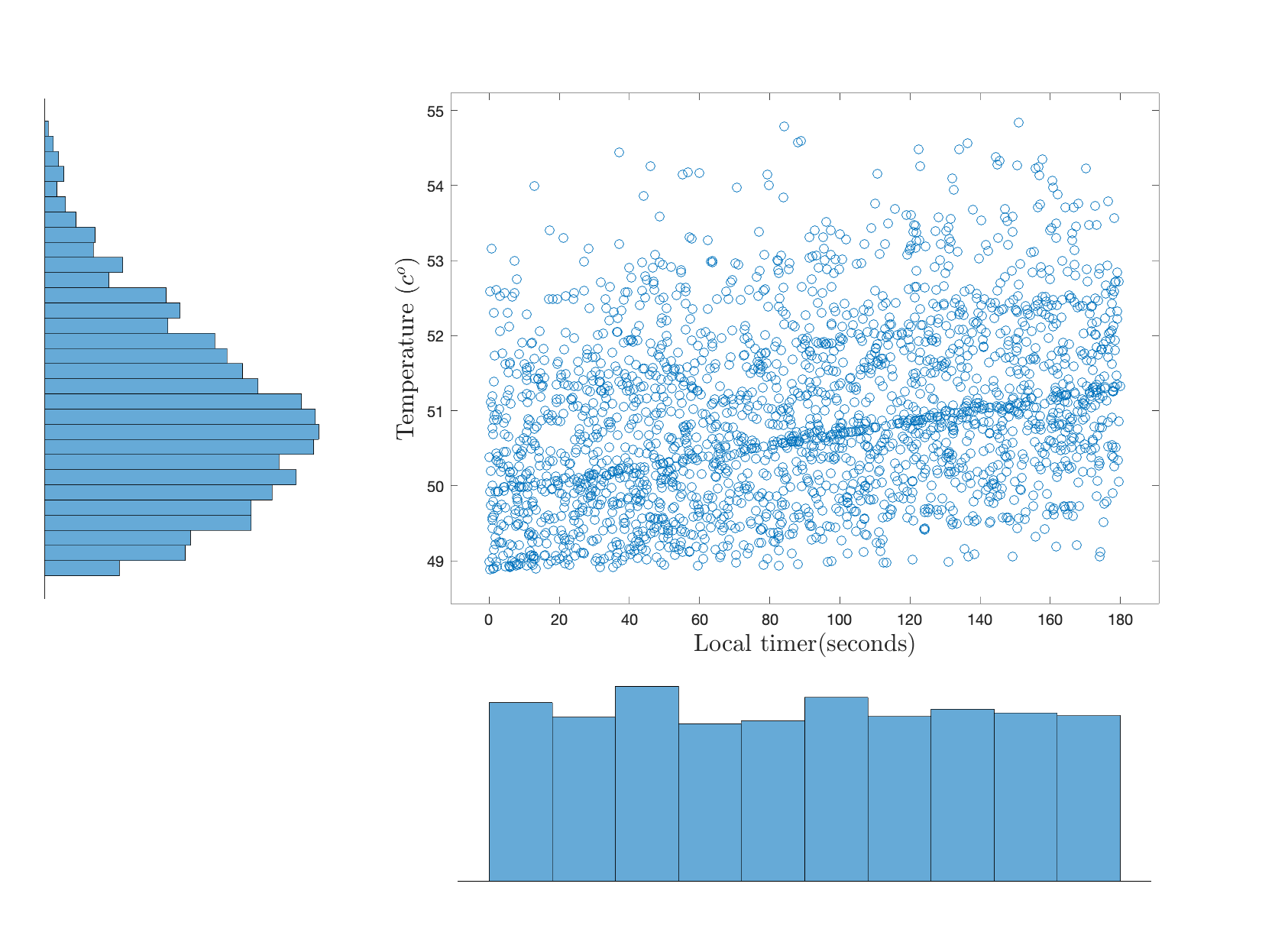}
		\caption{The state of local timers and temperature of ON devices before the disturbance. }
		\label{fig:localtimer}
	\end{figure}

    
	To guarantee quality of service (QoS) for the end-consumers, packetized-enabled DERs are endowed with opt-out logic~\cite{duffaut2018PSCC}. This means that if the DER's dynamic state exceeds predefined comfort limits, the device will exit PEM temporarily and return to its default control logic (e.g., TCL) until the dynamic state is returned to within comfort limits. Since the dynamic state evolves much more slowly than the grid frequency and power deviations and the duration of frequency event is shorter than 30 seconds, the opt-out behaviors are negligible. In fact, since the duration of the frequency response is short, it is reasonable to assume that the DERs' energy states are fixed for the duration of the frequency event. 
	
	The timer state distribution is clearly a function of past aggregator packet acceptance rates. If the fleet's aggregate power is near the nominal power reference, which is given by nominal QoS (e.g., set-point temperature), then from~\cite{duffaut2020CDC}, the rates of requests and accepted requests are uniformly distributed. That is, $x[k] \approx \bar{x}$ for all bins. These conditions are representative of a packetized DER fleet tracking ISO regulation signals, such as tracking ISO-NE's energy neutral or PJM's Reg-D, and ensure that the net energy exchanges of the fleet are close to zero. 
	To make the scheme more responsive to frequency deviation, packets can be actively interrupted. The packet interruption is defined next.
	\begin{definition}{(Packet Interruption)}
    The interruption of a packet is the instantaneous termination of a packet before the end of its epoch length ($t_n < \delta$) due to local conditions.
    \end{definition}	
	With packet interruptions, PEM-related demand can be actively modified by appropriately interrupting packets. For example, for an under-frequency event, packet interruptions enable DERs to turn OFF devices based on their local timer and/or temperature state, which offer a simple scheduling mechanism for dynamically prioritizing which devices turn OFF at what frequency as a contingency unfolds. This approach also limits potentially harmful synchronization effects.
	
	The proposed decentralized frequency control scheme is provided in the next Section, where PEM-enabled DERs leverage information about their local packet duration and timer state.

	\section{Decentralized frequency control for PEM }\label{sec:decentralized}
	
	To design a decentralized control policy for PEM, one needs to consider the local data/measurements available to each DER $n$ at node $i$: 1) frequency, $f_i[k]$, 2) temperature, $T_n[k]$, and 3) timer state, $t_n[k]$. These enable each DER to make local ON/OFF decisions. In that context, each DER decides whether or not requests should be sent to the aggregator. In other words, when a PEM-enabled DER senses that the frequency deviation exceeds a predefined dead-band, the DER blocks requests locally and switches to a decentralized control policy. This is detailed next. 
	
    A dead-band around $f_0$ is set in order to define the transmission reliability criteria. Within this dead-band, $\Delta f_{\text{db}}$, conventional PEM is used to provide ancillary and whole-sale energy market services. However, deviating beyond the dead-band represents reliability concerns, so PEM-enabled DERs switch to decentralized control policy to actively support primary frequency response.
    
    For the sake of simplicity, this paper focuses on EWHs. As mentioned earlier, the dynamic state of devices evolves slowly and does not change during the frequency event. In addition, the aggregator knows the timer state distribution. Therefore, one can estimate how the number of packet interruptions affect the damping provided by PEM devices ($D_{\text{PEM}}$). The analysis performed herein assumes a fleet whose aggregate energy dynamics (e.g., distribution of device temperatures) are not changing much with time. Under this assumption, the timer states distribution follows a uniform distribution, which simplifies analysis as explained next. Extending the analysis of decentralized PEM to arbitrary distribution of timer states is straightforward since the aggregator have the distribution of timer states at each time. 
    
    A naive initial approach to reduce demand implies to automatically accept or reject all packet requests during frequency disturbances. However, this does not offer a sufficient change in demand to affect the frequency response since such approach relies on slow packet completion rates. Another overly simplistic approach consists of interrupting all of the timers simultaneously when any frequency deviation occurs. This triggers a step change in demand that ignores the frequency's evolution and can cause system instabilities~\cite{mirosevic2015ESARS} if the share of DERs in the power system is significant. Therefore, it is necessary to prioritize devices so that the ones with higher timer or temperature turn OFF first in under-frequency events. In fact, one needs to dynamically interrupt the packets to reduce demand and have a meaningful effect on the frequency.
	In the proposed method, a packet interruption threshold is assigned to the local timer based on the local frequency measurement as shown in Fig. \ref{fig:diflocaltimer}. That is, when the magnitude of frequency deviation is smaller than $\Delta f_{\text{db}}$, no control action is needed. The design of this deadzone depends on power system reliability requirements defined by transmission operators. For frequency deviations between $\Delta f_{\text{db}}$ and $\Delta f_{\text{max}}$, $\eta \in [0,1]$. If the frequency deviation is larger than $\Delta f_{\text{max}}$, the value of $\eta$ will remain constant at $\eta_{\max} \in [0,1]$. Therefore, the only local design parameters are $\eta_{\text{max}}$, $\Delta f_{\text{db}}$, and $\Delta f_{\text{max}}$. For $\eta_{\text{max}}\approx 1$, the aggregate decentralized PEM response is more aggressive due to more interruptions for a given frequency deviation. A linear function is used for $\eta (\Delta f)$ which will result in an aggregate response that adds equivalent, constant damping to the system. The proposed local control law is given by:
	\begin{align}\label{eq:eta}
	\eta(\Delta f)= 
    	\begin{cases} 
    	\mbox{0,} & \mbox{$ \Delta f_{\text{db}} < \Delta f[k]$}\\ 
    	\mbox{$ \frac{\Delta f[k]-\Delta f_{\text{db}}}{\Delta f_{\text{max}}-\Delta f_{db}} \eta_{\text{max}}$,} & \mbox{$\Delta f_{\text{max}}\le \Delta f[k] \le \Delta f_{\text{db}}$}\\ 
    	\mbox{$\eta_{\text{max}}$,} & \mbox{$\Delta f[k] < \Delta f_{\text{max}}.$} 
    	\end{cases} 
	\end{align}
	\begin{figure}[t]
		\centering
		\includegraphics[width=0.5\linewidth]{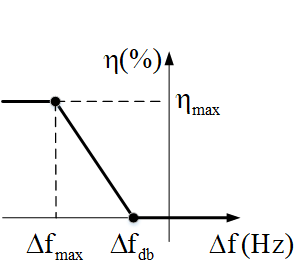}
		\caption{Illustrating the proportion of interrupted devices based on any locally measured frequency deviation from nominal, $\Delta f[k]:=f[k] - f_0$.}
		\label{fig:diflocaltimer}
	\end{figure}
	The aggregate effect of the control law~\eqref{eq:eta} in a simple two-area system is shown in fig.~\ref{fig:timerbased} for different values of $\eta_{\text{max}}$. It can be seen that larger $\eta_{\text{max}}$ results in more damping. When $\eta_{\text{max}}$ is zero, none of the ON devices are interrupted. In this case, all of the requests are rejected locally and no new energy packet request is sent to the aggregator. Thus, consumption decreases with a constant rate equal to packet completion rate which is $P^\text{rate}\bar{x}/\Delta t$. 
	Observe in Fig.~\ref{fig:timerbased} that the packet completion rate is relatively slow and have small impact on ROCOF (blue curve). Increasing packet interruption leads to a more sudden and larger drop in PEM demand for under-frequency events, which improves the ROCOF, maximum frequency deviation (also called the nadir point), and final frequency deviation. The demand starts to decrease, $\Delta t$ seconds after the occurrence of disturbance. Table~\ref{t.1} shows ROCOF, $\Delta f_{\text{nadir}}$ and final frequency deviation for different values of $\eta_\text{max}$, which illustrates the effectiveness of the proposed prioritization scheme. As seen in Table~\ref{t.1}, interrupting more energy packets (larger $\eta_{\text{max}}$) improves the frequency response of the system.\\ 
\begin{table}[htbp]
	\centering
	\caption{Characteristics of frequency response for different $\eta_{max}$}	
	\begin{tabular}{rrrr} 
	\toprule
		$\eta_{max}$ & ROCOF& $\Delta f_{\text{Nadir}}$ & $\Delta f_{\infty}$\\
		  &(mHz/sec)&(mHz) & (mHz) \\
		  \hline
		0 & 104 & 83 & 46  \\
        0.33 & 94 & 75 & 42  \\
        0.67 & 86 & 69 & 39 \\
        1 & 81 & 64 & 36 \\
        \bottomrule
	\end{tabular}
	\label{t.1}
\end{table}

	In a conventional power system, after a loss of generation, the frequency decreases rapidly until it achieves a minimum value and then it partially recovers due to the remaining generators' primary droop controllers and system damping. In this paper, $\Delta f_{\text{nadir}}$ is defined as the frequency deviation at nadir point. According to~\eqref{eq:eta}, when the frequency achieves its nadir point, $\eta$ achieves its maximum value, which means that the largest proportion of devices are interrupted at this time. After the frequency recovers away from the nadir point, $\eta$ decreases, but this does not change the number of interrupted devices since no new devices are turned ON after the nadir point\footnote{Future wor k will focus on using batteries to discharge or including the ability to switch devices back ON during the frequency deviations.}. In other words, the frequency deviation at the nadir point, provides the largest $\eta$, which determines the damping provided by the PEM fleet.


    	\begin{figure}
		\centering
		\includegraphics[width=1\linewidth]{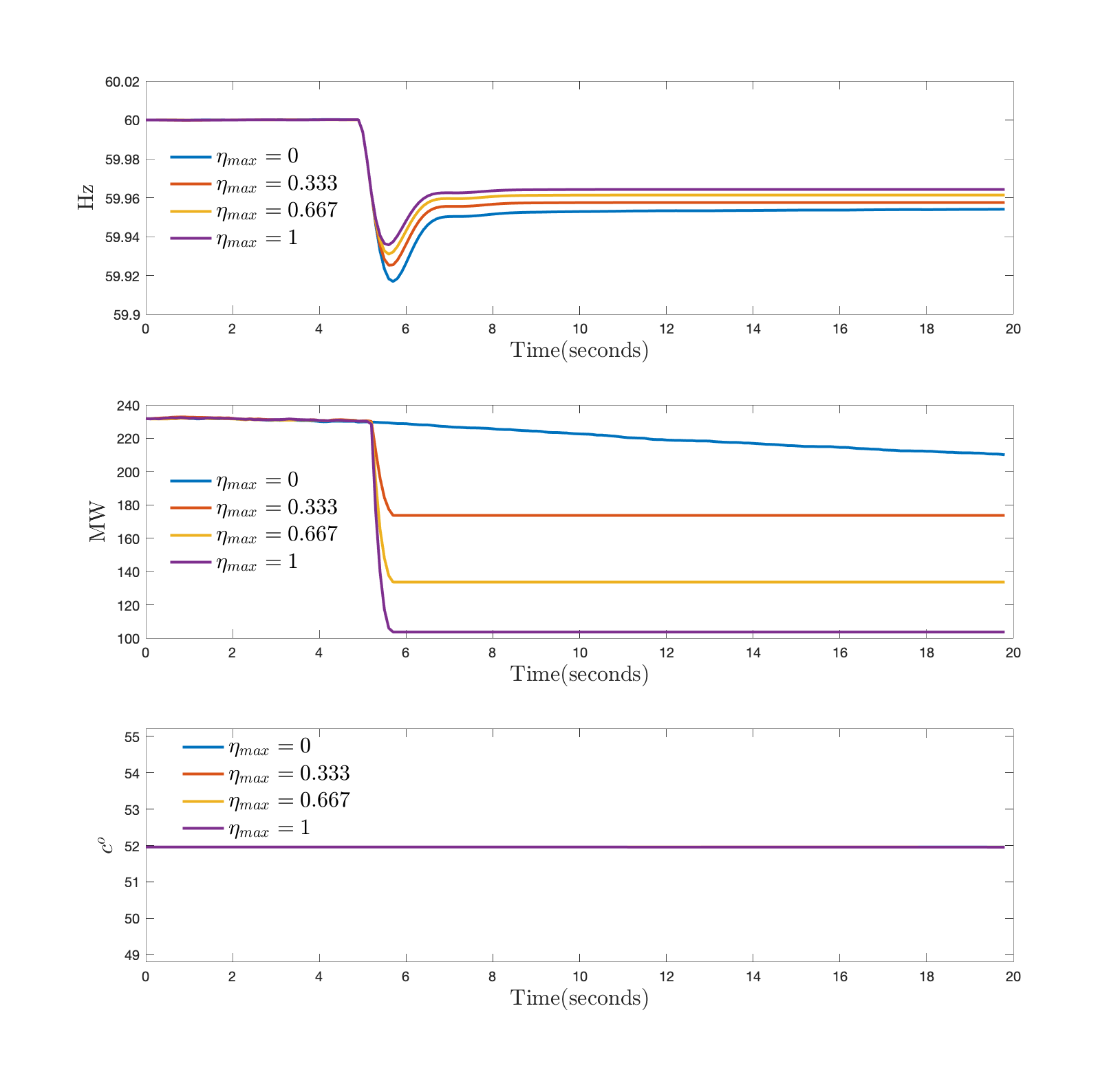}
		\caption{Aggregate power, frequency and average temperature for different values of $\eta_{max}$ for 400,000 DERs. A 500 MW drop in total generation occurs at $t=5$ seconds.}
		\label{fig:timerbased}
	\end{figure}
    
	Clearly, the decentralized PEM scheme can provide damping as seen in Fig. \ref{fig:timerbased}. This damping is achieved with no coordination between DERs and aggregator. Being able to estimate the damping available from a fleet of DERs would be valuable for grid operators and market participants interested in fast frequency response (FFR) markets~\cite{aemo2017}. The next section provides an accurate online estimate of the equivalent damping provided by a fleet of DERs operating under~\eqref{eq:eta}.

\section{Predicting the aggregate response}\label{sec:predict}
    
    In this section, the equivalent damping provided by a fleet of DERs operating under the decentralized control law in~\eqref{eq:eta} is estimated analytically. In addition,~\eqref{eq:eta} is augmented to combine both local timer and temperature information in the fully decentralized packet interruption scheme and estimate the resulting equivalent damping. While the PEM aggregator has direct access to the distribution of timer states, the temperature distribution required for the online estimate must be obtained indirectly with a state estimator~\cite{duffaut2020PES,duffaut2018PSCC}.
    
    Next, an analytical estimate of the equivalent damping for just the aggregator's timer state distribution is provided. 
   
   \subsection{Timer-based prioritization}\label{subsec:timerthreshold} 
	
	In order to predict the equivalent PEM damping, the aggregator leverages available real-time information about the distribution of timer states. The aggregator makes use of the following simplifying assumptions that are reasonable in a practical setting, to estimate the damping. 
    
\begin{assumption}\label{ass:1}
    The DER population is large enough and operates near nominal power so that the timer bins are well-approximated by $\bar{x}[k]$ in \eqref{eq:xbar}.  
\end{assumption}

\begin{assumption}\label{ass:2}
    The average number of devices at each bin $\bar{x}_{\rm{nom}}$ does not change with time. That is, $\bar{x}_{\rm{nom}} \approx \bar{x}[k]$.

\end{assumption}

    \begin{assumption}\label{ass:3}
    The frequency response event duration is less than 30~seconds. That is, it is assumed that the nadir point is such that $\eta(\Delta f_{\rm{nadir}}) > \frac{30}{\delta}$. This implies that one can neglect the effect of packet completions rate.
    \end{assumption}
    For example, consider a PEM system with $\delta=180$s and frequency event with a nadir such that $\eta(\Delta f_{\text{nadir}})=0.9$. Then, all devices with timer states $t_n > 18$ are interrupted and natural packet completions will not occur for 162 seconds and can, therefore, be neglected.
     The analytical estimate is embodied by the following theorem.
\begin{theorem}\label{thm:1}	
		Let $\delta$, $\Delta f$ and $P^{\rm{rate}}$, and $n_{\rm{p}}$ be fixed for a DER fleet with decentralized control policy~\eqref{eq:eta} and chosen  $\eta_{\rm{max}}$, $\Delta f_{\rm{max}}$ and $\Delta f_{\rm{db}}$. Under assumptions~\ref{ass:1},~\ref{ass:2} and~\ref{ass:3}, the PEM fleet responds to frequency deviations with an equivalent damping of 
     \begin{align}\label{eq:damp_linear}
		D_{\rm{PEM}}
		&=
		\left \{
	\begin{array}{rl}
    	0, &  \!\!\!\! \Delta f_{\rm{db}} < \Delta f_{\rm{nadir}} \\ 
    	P^{\rm{rate}}\frac{\eta_{\rm{max}}n_{\rm{p}} \bar{x}_{\rm{nom}}}{\Delta f_{\rm{db}}-\Delta f_{\rm{max}}}, & \!\!\!\! \Delta f_{\rm{nadir}}\in [\Delta f_{\rm{max}}, \Delta f_{\rm{db}}].\\ 
    	P^{\rm{rate}}\frac{\eta_{\rm{max}}n_{\rm{p}} \bar{x}_{\rm{nom}}}{\Delta f_{\rm{db}}-\Delta f_{\rm{nadir}}}, & \!\!\!\! \Delta f_{\rm{nadir}} < \Delta f_{\rm{max}} 
	\end{array}
    	\right.
    \end{align}
  \end{theorem}
\begin{proof}
    The proof is by construction. A PEM-enabled DER can either naturally complete or interrupt its packet. From~\eqref{eq:eta}, if $\Delta f_{\text{nadir}} > \Delta f_{\text{db}}$, then no device is interrupted and the PEM fleet is not responsive to the frequency, so the equivalent damping is zero. 
    
    For larger frequency deviations, assumption~\ref{ass:3} ensures that one only has to consider packet interruptions. Thus, the total change in PEM load for system area $j$ is described as:
	\begin{align}\label{eq:accurate_dpm}
	\Delta P^{\text{PEM}}_{j}(\Delta f_{\text{nadir}}) = P^{\text{rate}}\sum_{i=\left\lceil (1-\eta\Delta f_{\text{nadir}})n_\text{p}\right \rceil } ^{n_\text{p}} x_i[k],
	\end{align}
	where $x_i[k]$ is the $i^{\text{th}}$ entry of an arbitrary timer states distribution $x[k]$. From Assumption~\ref{ass:1}, the total number of interruptions can be approximated by multiplying total number of ON devices and $\eta(\Delta f_\text{nadir})$. Therefore,~\eqref{eq:accurate_dpm} can be rewritten as follows:
	\begin{align}\label{eq:reducedpower2}
	\Delta P^{\text{PEM}}_{j}(\Delta f_{\text{nadir}}) \approx P^{\text{rate}}\eta(\Delta f_{\text{nadir}}) (\mathbf{1}^{\top}_{n_p} x[k]) ,
	\end{align}

	 In addition, from Assumption~\ref{ass:2}, $\mathbf{1}^{\top}_{n_p} x[k]\approx N_b \bar{x}_{\text{nom}}$. For $\Delta f_{\text{nadir}} \in [\Delta f_{\text{max}},\Delta f_{\text{db}}]$,~\eqref{eq:eta} gives $\eta(\Delta f_{\text{nadir}})=\left( \frac{(\Delta f_{\text{nadir}} - \Delta f_{\text{db}})\eta_{\text{max}}}{\Delta f_{\text{max}}-f_{\text{db}}}\right)$ and substituting this into~\eqref{eq:reducedpower2} yields
	\begin{align*}
	\Delta P^{\text{PEM}}_{j}(\Delta f_{\text{nadir}}) =P^{\text{rate}}\left(\frac{(\Delta f_{\text{nadir}} - \Delta f_{\text{db}})\eta_{\text{max}}}{\Delta f_{\text{max}}-f_{\text{db}}}\right) N_\text{b} \bar{x}_{\text{nom}}.
	\end{align*}
	Since this change in power occurs over frequency deviation $\Delta f_\text{nadir}-\Delta f_\text{db}$, the equivalent damping in~\eqref{eq:damp_linear} is obtained.
	
	Finally, for $\Delta f_{\text{nadir}} < \Delta f_{\text{max}}$,~\eqref{eq:eta} saturates and $\eta(\Delta f_{\text{nadir}})=\eta_{\text{max}}$. Then,  $\Delta P^{\text{PEM}}_{j}(\Delta f_{\text{nadir}}) \approx P^{\text{rate}} \eta_{\text{max}} N_\text{b}\bar{x}_{\text{nom}} .$ The equivalent damping is then 
	 $\frac{\Delta P^{\text{PEM}}_{j}(\Delta f_{\text{nadir}})}{\Delta f_{\text{nadir}}-\Delta f_{\text{db}}}$, which results in~\eqref{eq:damp_linear}.
	This concludes the proof. 
    \end{proof}


	\begin{remark}
     Theorem~\ref{thm:1} allows the PEM fleet to be modeled as a proportional controller with gain $ D_{\rm{PEM}}$. Note also, from~\eqref{eq:accurate_dpm}, that Theorem~\ref{thm:1} can be extended to estimate $D_{\rm{PEM}}$ in real-time for an arbitrary distribution of timer states. Nevertheless, $D_\text{PEM}$ will no longer be constant for all $\Delta f_{\rm{nadir}}\in [\Delta f_{\rm{max}}, \Delta f_{\rm{db}}]$. This represents ongoing work.
	\end{remark}
	Next, the proposed scheme is implemented on a two-area system to verify its performance.

	\section{Numerical validation with Two-area system}\label{sec:results}
	
	Consider the so-called two-area model \cite{anderson2012}. When a fleet of DERs under PEM interact with this system as in Fig.~\ref{fig:2area}, the steady state frequency deviation and damping  are  given, respectively, by
	\begin{align}\label{eq:damping}
	\nonumber \Delta f_{\infty}& =\frac{\Delta P_{G}}{D^{\text{actual}}_{\text{PEM}}+2(D_j+\frac{1}{R_j})} \mbox{   and  }\\
	D^{\text{actual}}_{\text{PEM}} & =\frac{\Delta P_{G}}{\Delta f_{\infty}}-2\left(D_j+\frac{1}{R_j}\right),
	\end{align}
	where $D_j$ and $R_j$ are the damping in MW/Hz and droop constant in Hz/MW in each area $j$. Here, it has been assumed that both areas have equal damping and inertia. PEM loads and generation drop are in area~2. The two-area system parameters and simulation setup with respect to PEM are presented in table~\ref{t.3}. Fig.~\ref{fig:2area} depicts the two area system interacting with PEM loads. Equation~\eqref{eq:damping} is used to compute the actual damping of the system in simulations. In what follows, $D^{\text{actual}}_{\text{PEM}}$ will be compared against~\eqref{eq:damp_linear} in Theorem~\ref{thm:1}.  

 	\begin{figure}
 		\centering
 		\includegraphics[width=1\linewidth]{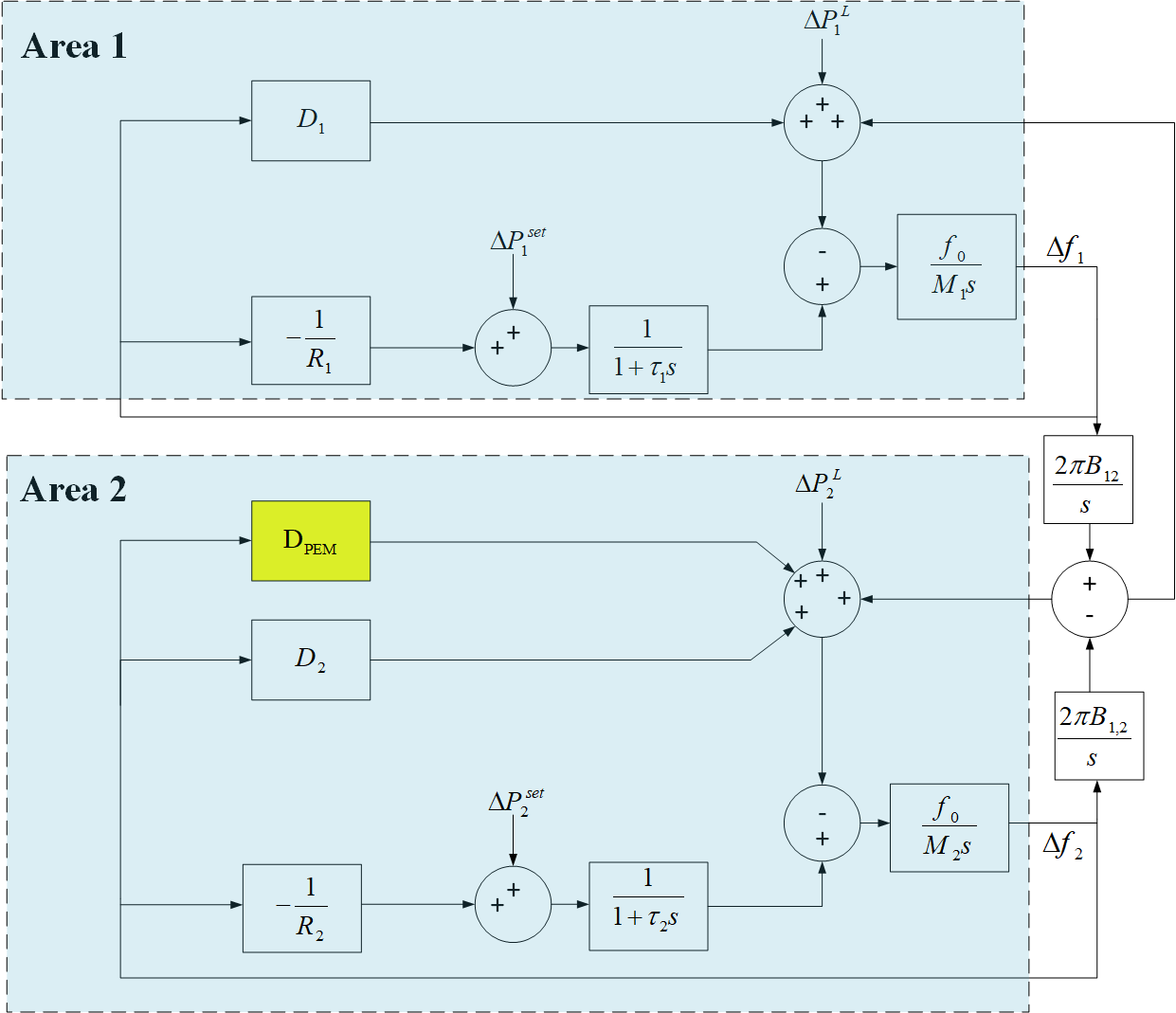}
 		\caption{Block diagram of primary frequency control of a two-area power system}
 		\label{fig:2area}
 	\end{figure}

\begin{table}[htbp]
	\centering
	\caption{Simulation Parameters}	
	\begin{tabular}{ll} \toprule
	
		Parameter & Value\\ 
		\hline
		H  & 5 seconds\\
		$f_{0}$ & 60 Hz  \\
		$D_j$ &  200 MW/Hz\\ 
		$R_j$ &  $\frac{1}{5000}$ Hz/MW\\
		$\Delta f_{\text{db}}$& 20 mHz  \\
		$\Delta f_{\text{max}}$& 100 mHz \\
		$T^{set}_{n}$& 52 $C^{o} $ \\ 
		$T^{max}_{n}$& 55.2 $C^{o} $ \\
		$T^{min}_{n}$& 48.8 $C^{o} $ \\
		$\Delta t$ & 100 ms  \\
		Simulation Time & 20 seconds  \\
		MTTR & 3 min \\
		Fleet size & 400,000 \\
		Disturbance & -500MW @ $t=5$ sec \\
		Epoch & 3 min\\ \bottomrule
	\end{tabular}
\label{t.3}
\end{table}
	
	\subsection{Equivalent damping of PEM loads}
	As seen in the previous subsection, when energy packets are interrupted ($\eta_{max} \neq 0$), the number of natural completions during the disturbance is negligible. Therefore,~\eqref{eq:damp_linear} can be used to calculate the estimated damping of the population of PEM loads. Then, the PEM fleet operating under decentralized control policy from~\eqref{eq:eta} can be modeled as a simple, lumped proportional frequency-responsive demand $D_{\text{PEM}}$. 
	Fig.~\ref{fig:comparison}, compares simulation results using PEM loads against equivalent proportional controller. It can be seen that for $\eta_{\text{max}}=0$, the error is higher compared to other cases since packet completion rates are not negligible. The Root Mean Square Error (RMSE) for different values of $\eta_{\text{max}}$ and relative error in damping estimation is presented in table~\ref{t:2}. The estimation error is calculated as $100\times\frac{D_{\text{PEM}}(\Delta f_{\text{nadir}})-D^{\text{actual}}_{\text{PEM}}}{D^{\text{actual}}_{\text{PEM}}}$.
	The results show that the relative accuracy of the damping estimate improves as the frequency deviation increases, however, for all estimates, the resulting frequency response in Fig.~\ref{fig:comparison} matches closely with RMSE$< 1.2$mHz. The estimate is made solely based on distribution of timers, which is available to the aggregator. No online measurements or communication are required for DERs to respond to frequency.
	
	\begin{table}[htbp]
	\centering
	\caption{Accuracy of online estimation of Damping}	
	\begin{tabular}{rrr} 
	\toprule
		$\eta_{max}$ & RMSE & Estimation \\
		             &(mHz) & error (\%) \\
		  \hline
        0 & 1.1 & -   \\
        0.33 & 0.6 & 12.4   \\
        0.67 & 0.6 & 5.5  \\
        1 & 0.5 & -0.5\\
        \bottomrule
	\end{tabular}
	\label{t:2}
\end{table}

	\begin{figure}
		\centering
		\includegraphics[width=1\linewidth]{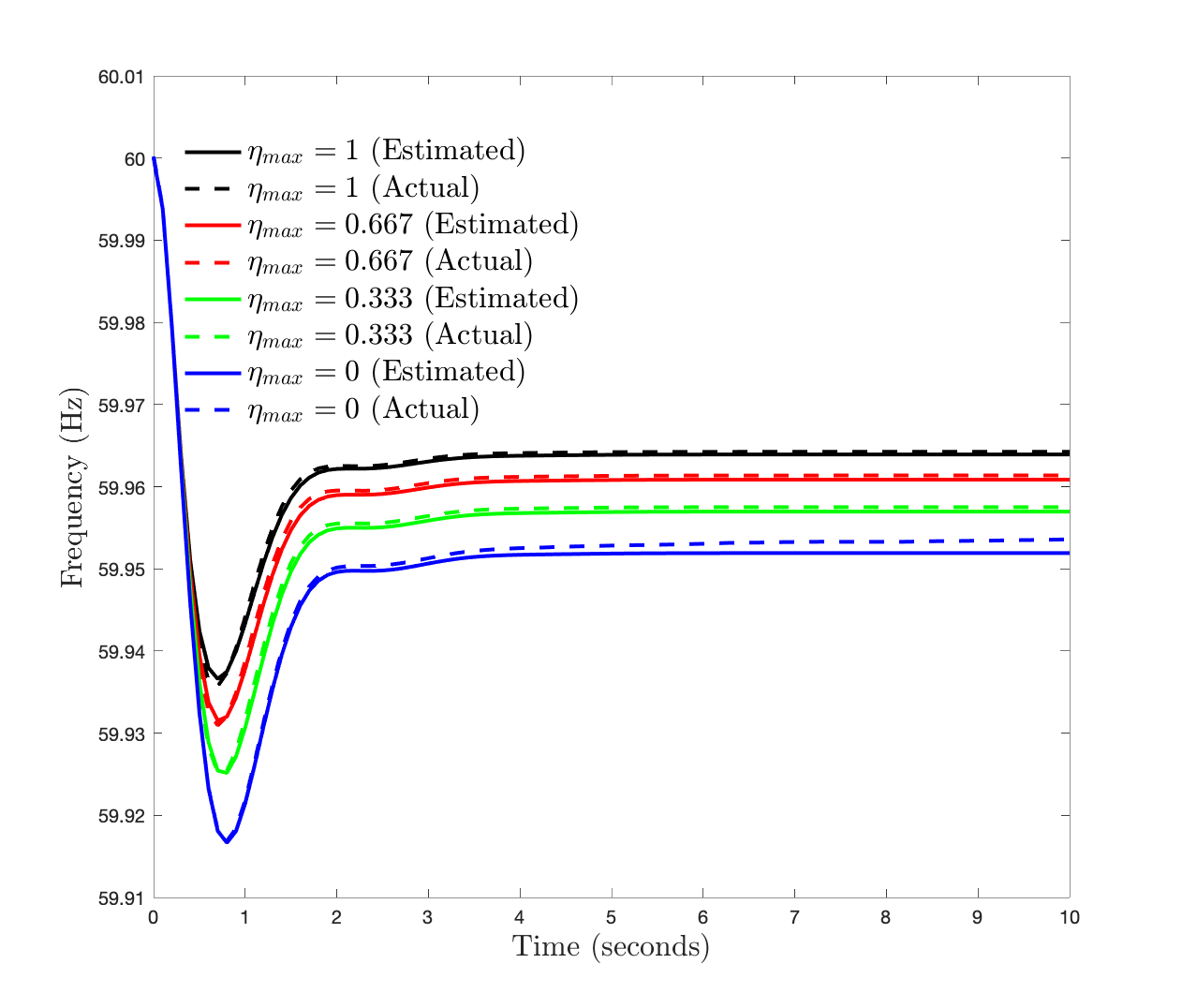}
		\caption{Estimation of frequency response for different values of $\eta_{max}$}
		\label{fig:comparison}
	\end{figure}
	
	\section{conclusion}\label{sec:conclusion}
	A packet-based model for primary frequency control was designed. It was shown through simulations that proposed method was able to provide fast frequency response and improved ROCOF and maximum frequency deviation in a fully decentralized manner. In addition, an equivalent damping for the fleet of DERs was calculated based on online measurements available to aggregator. Future work includes studying the combined temperature and timer prioritization and generalizing this work for over-frequency events by modifying the local control law together with using batteries.

\bibliographystyle{IEEEtran}
\bibliography{ref2}


\end{document}